\documentclass[letterpaper]{article}
\usepackage[preprint]{aaai2027}
\usepackage[hyphens]{url}
\usepackage{graphicx}
\usepackage{natbib}
\usepackage{caption}
\usepackage{subcaption}
\usepackage{booktabs}
\usepackage{multirow}
\usepackage{amsmath}
\usepackage{amssymb}
\usepackage{amsthm}
\usepackage{placeins}
\usepackage{xcolor}
\usepackage[most]{tcolorbox}
\newtheorem{theorem}{Theorem}
\newtheorem{proposition}[theorem]{Proposition}
\newtheorem{lemma}[theorem]{Lemma}
\newtheorem{corollary}[theorem]{Corollary}
\newtheorem{definition}{Definition}
\NewDocumentEnvironment{remark}{o}{%
  \begin{tcolorbox}[
    enhanced,
    breakable,
    colback=black!3,
    colbacktitle=black!9,
    colframe=black!28,
    boxrule=0.45pt,
    arc=2mm,
    outer arc=2mm,
    left=3mm,
    right=3mm,
    top=2mm,
    bottom=2mm,
    toptitle=1.3mm,
    bottomtitle=1.3mm,
    boxsep=0pt,
    before skip=7pt,
    after skip=7pt,
    fonttitle=\bfseries,
    coltitle=black,
    title={Remark},
    title after break={Remark (continued)}
  ]%
}{%
  \end{tcolorbox}%
}

\newcommand{\minex}{\textsc{MinEx}}
\newcommand{\cleanup}{\textsc{CleanUp-lite}}

\newcommand{\cset}{\mathcal{C}}

\usepackage{float}
\usepackage{multicol}
\usepackage{enumitem}

\title{After Cooperation Is Learned: Gradient Routing and Optimizer-Dependent Maintenance in Multi-Agent Reinforcement Learning}
\author{
Chaoyuan Hao, Wentao Yue, Tianyou Lai,\\
Hongji Li, Jiayi Zhou, Qingyu Mao, Qilei Li
}
\affiliations{}

\newtheorem{suppTheorem}{Theorem}
\newtheorem{suppProposition}[suppTheorem]{Proposition}
\newtheorem{suppLemma}[suppTheorem]{Lemma}
\newtheorem{suppCorollary}[suppTheorem]{Corollary}
\theoremstyle{remark}
\newtheorem{suppRemark}[suppTheorem]{Remark}
\theoremstyle{plain}
\begin{document}

\maketitle
\begin{abstract}
Cooperative MARL is commonly evaluated through cooperation discovery from
  random initialization, leaving open whether continued optimization can
  destabilize learned cooperation. Actor--critic comparisons can also conflate
  critic presence with value gradients entering shared actor representations.
  We study \emph{cooperation maintenance}, defined as the survival of a
  behaviorally verified cooperative policy under continued training. We
  formulate maintenance as a right-censored event-time problem and compare
  matched warm starts: X0 allows value loss gradients to update shared actor
  features, X1 retains the critic while blocking those gradients, and X5 removes
  the learned critic as a critic-free reference. This isolates direct
  value-gradient access while controlling initialization, critic computation,
  and evaluation. Positive reward scaling preserves strategic preferences and
  equilibria while perturbing learning dynamics. Gradient audits confirm the
  intended routing pathways, and frozen-policy torso perturbations probe whether
  route-induced updates align with local cooperation boundaries. In
  confirmatory \minex{} and \cleanup{} experiments, higher scales selectively
  increase maintenance sensitivity in X0; X1 remains near the censoring ceiling,
  and X5 has no confirmed events in the tested settings. In \cleanup{},
  route-by-scale displacement is associated with reduced local cooperation
  margins; \minex{} shows a weaker, optimizer-dependent effect. These results
  identify a conditional, scale-sensitive maintenance risk associated with
  direct value-gradient routing rather than a universal failure of critics.
\end{abstract}

\vspace{-0.1in}
\section{Introduction}

Most cooperation research in multi-agent reinforcement learning (MARL) asks
whether agents can \emph{discover} a cooperative policy.  Sequential social
dilemmas and broad evaluation suites have made this question increasingly
precise \citep{leibo2017multi,leibo2021scalable,guo2026socialjax}.  Deployment,
however, raises a different question: \textbf{after cooperation has been
learned, can continued optimization erase it?}  Endpoint performance from random
initialization cannot isolate this maintenance problem because failure may
instead reflect insufficient exploration, failed coordination, or equilibrium
selection.

We focus on one mechanism that can produce such post-discovery failure:
\textbf{interference through shared actor--critic representations}.  Although
such sharing allows policy and value learning to exploit common features, it
also gives value optimization a direct route into the representation that
determines behavior.  The critic may improve credit assignment, yet its
gradients can simultaneously move the actor away from cooperation.  Critic
presence and value-gradient access are therefore distinct factors.  Because
representation sharing and objective separation are active design choices in
modern policy-gradient methods
\citep{cobbe2021phasic,raileanu2021decoupling,andrychowicz2021matters},
comparisons between unrelated actor--critic algorithms cannot isolate this
route from differences in architecture, initialization, or optimization.

To address both the evaluation and identification gaps, we define
\emph{cooperation maintenance} as the survival of a behaviorally verified
cooperative policy under continued learning and evaluate it from a common
cooperative checkpoint.  All runs begin from byte-identical actors and follow
one of three matched training arms: X0 allows the value loss to update features
shared with the actor; X1 retains the same critic and its forward computation
as X0 but blocks the value gradient at the shared features; and X5 removes the
learned critic.  With the arms defined, the X0--X1 contrast isolates direct
value-gradient access while controlling the initial actor, critic computation,
and evaluation, whereas X5 serves as a critic-free reference for failures that
persist without value learning.  We then apply positive reward scaling as a
\emph{game-preserving intervention}: it leaves strategic preferences, best
responses, equilibria, and linear incentive constraints unchanged while
perturbing value targets and optimizer dynamics.  Frozen evaluations along
each learning trajectory record first failure, sustained failure, and
recovery.  Treating these outcomes as right-censored event times distinguishes
transient deviations from persistent collapse and includes runs that remain
cooperative through the evaluation horizon.

Across two structurally distinct environments, the experiments reveal a
conditional, scale-sensitive maintenance risk concentrated in X0 rather than
a universal failure of critics.  In the confirmatory cells, X0 becomes more
fragile as reward scale increases, while X1 remains substantially more robust
and the critic-free X5 reference has no confirmed events.  Direct gradient
audits verify that the intended route is active only in X0, and local
behavioral-margin tests show clear adverse alignment in \cleanup{} but weaker,
contrast-dependent evidence in \minex{}.  Complementary optimizer settings
yield either all-arm survival or failure even without a critic, showing that
the observed routing-specific risk depends on both optimizer and environment.

Our contributions are:
\begin{itemize}
    \item We formulate cooperation maintenance as a right-censored temporal
    problem, separating post-discovery behavioral survival from exploration,
    coordination, and equilibrium-selection failures.
    \item We introduce a matched X0/X1/X5 warm-start design that isolates
    direct value-gradient access to actor representations, and combine it with
    game-preserving reward scaling, exact fixed-batch routing results, and
    production-path gradient audits.
    \item We provide fresh-seed confirmatory evidence across \minex{} and
    \cleanup{}, together with local behavioral-margin and optimizer-boundary
    diagnostics that separate outcome association from behavioral mechanism
    and delimit the claim.
\end{itemize}

\begin{figure*}[t]
\centering
\includegraphics[width=0.98\textwidth]{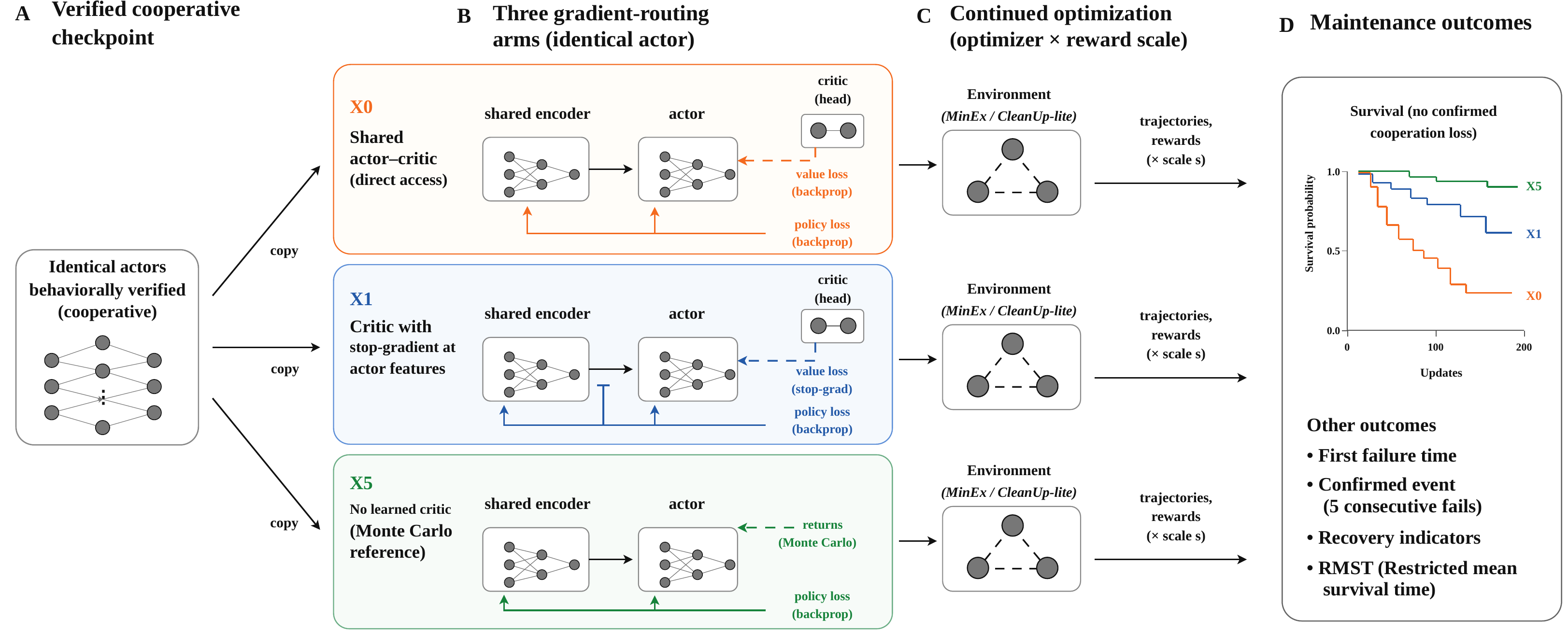}
\caption{\textbf{From discovery to maintenance.}
\textbf{(A)} Start from a verified cooperative actor.
\textbf{(B)} X0 passes value gradients into the shared torso, X1 blocks them
while retaining the critic, and X5 removes the critic.
\textbf{(C)} Continued optimization varies optimizer and strategically
equivalent reward scales.  \textbf{(D)} Frozen evaluations track cooperation
loss and recovery.}
\label{fig:concept}
\end{figure*}

\vspace{-0.1in}
\section{Related Work}
\noindent\textbf{Cooperation and social dilemmas in MARL.}
Sequential social dilemmas formalize cooperation as a temporally extended
policy rather than a one-shot action \citep{leibo2017multi}.  Recent suites
broaden evaluation across substrates, unseen partners, and efficiently
implemented social dilemmas
\citep{leibo2021scalable,guo2026socialjax}.  Cooperative MARL has also expanded
the mechanisms used to obtain strong joint behavior, including
reward-assisted learning and sequence models for coordinated action
\citep{wang2022individual,wen2022multi}.  These lines of work sharpen the
question of whether cooperation can be discovered, transferred, or
generalized.  Our estimand starts after acquisition: it measures how long a
specific, verified checkpoint remains within an operational cooperation set
while optimization continues.  Conditioning on prior success removes initial
exploration and equilibrium selection from the primary comparison, allowing
post-discovery behavioral loss to be studied directly.

This focus is related to, but different from, recent work on plasticity loss.
Deep-RL training can produce primacy bias, dormant representations, reduced
capacity, and a declining ability to incorporate experience
\citep{lyle2022understanding,nikishin2022primacy,sokar2023dormant,
abbas2023loss,dohare2024loss}.  Resets and plasticity injection
can restore learning
\citep{d2022sample,nikishin2023deep}.  We fix the task and strategic
ordering, then test whether updates erase successful collective behavior.
Plasticity loss impairs acquisition; maintenance loss can occur while the
actor continues changing.  Both show that current performance does not
certify later behavior.

\noindent\textbf{Actor-Critic interference and representation sharing.}
Actor--critic methods are sensitive to implementation and optimization choices
\citep{schulman2017proximal,andrychowicz2021matters,yu2022surprising}.
PPG separates policy and value phases to limit objective interference
\citep{cobbe2021phasic}; IDAAC uses separate policy and value representations to
improve generalization \citep{raileanu2021decoupling}.  In cooperative MARL,
centralized and structured critics can coordinate policy improvement, but
their factorization, parameter sharing, and dependency structure affect what
is learned \citep{NEURIPS2021_65b9eea6,pmlr-v162-fu22d,chen2023context}.  Thus,
critics are not interchangeable components.  We isolate one direct
value-gradient route and test its effect on learned-cooperation survival.

Parameter sharing improves sample efficiency but couples representations and
updates \citep{christianos2021scaling}.  Cooperative PPO
is a strong baseline when its implementation is carefully controlled
\citep{yu2022surprising}.  MARL benchmarks emphasize matched implementations,
multiple seeds, uncertainty,
and transparent aggregation
\citep{papoudakis2021benchmarking,gorsane2022towards,ellis2023smacv2,
JMLR:v25:23-1612}.  Accordingly, our arms share the architecture and initial
actor, use paired seeds where appropriate, fix the endpoints in advance, and
report complete separation.

\noindent\textbf{Reward scale and reliable deep-RL evaluation.}
Positive reward scaling preserves fixed-policy strategic ordering, but
deep-learning updates are not generally scale invariant.  PPO practice
includes reward, value, or advantage normalization, clipping, and adaptive
optimization, which can introduce scale dependence
\citep{andrychowicz2021matters,yu2022surprising}.  We exploit rather than
eliminate this gap: scaling preserves the strategic problem while probing the
optimizer and direct value route.

Reliable RL evaluation requires more than final means from a small number of
runs \citep{agarwal2021deep,gorsane2022towards}.  Our primary outcomes
are event-time and RMST, supported by continuous behavior, raw event counts,
and paired-seed analyses.  These summaries remain interpretable when an arm
has no observed events.

\vspace{-0.1in}
\section{Problem Formulation}

\subsection{Discovery and Maintenance}

Let $\mathcal{G}$ be a Markov game, $\boldsymbol{\pi}$ a joint policy, and
$\cset_{\mathcal{G}}$ the set of policies satisfying a preregistered,
environment-specific cooperation criterion.

\begin{definition}[Cooperation discovery]
Given a training procedure initialized from a distribution
$\boldsymbol{\pi}_0\sim\nu_0$, discovery asks whether some learned policy
$\boldsymbol{\pi}_t$ enters $\cset_{\mathcal{G}}$.
\end{definition}

\begin{definition}[Cooperation maintenance]
Given a verified cooperative checkpoint
$\boldsymbol{\pi}_0\in\cset_{\mathcal{G}}$, maintenance asks for the survival
process
\begin{align}
S(t)=\Pr\big(&\boldsymbol{\pi}_{\tau}\in\cset_{\mathcal{G}}
\text{ without a confirmed event},\nonumber\\[-2pt]
&\forall\,0\leq\tau\leq t\big)
\end{align}
under continued learning.
\end{definition}

Discovery can fail because agents never coordinate.  Maintenance can fail
after exploration and equilibrium selection have been removed from the
comparison.  Conditioning on a successful checkpoint therefore changes the
failure variable: the outcome is no longer whether training ever reaches
cooperation, but how long subsequent updates preserve it.  This distinction
motivates the matched warm-start design and temporal endpoint below.

\vspace{-0.1in}
\subsection{Frozen Temporal Endpoint}

At fixed training intervals, we evaluate the current actor with learning
disabled and raw, unscaled environment rewards.  A checkpoint is successful
when all environment-specific criteria hold.  A \emph{confirmed maintenance
event} occurs at the fifth consecutive failed checkpoint after successful
true-step-0 evaluation.  The event time is the training time of the fifth
failure.  Runs without an event are right-censored at the frozen horizon
$\tau$.

We report:
\begin{align}
\text{RMST}(\tau) &= \int_0^\tau \widehat{S}(t)\,dt,\\
\text{AUC}_{J}(\tau) &= \frac{1}{\tau}\int_0^\tau J_{\mathrm{raw}}(t)\,dt,
\end{align}
along with first failure, recovery, final success, and continuous behavior.
Five consecutive failures reduce sensitivity to individual noisy evaluations;
the raster plots retain transient failures rather than erasing them.

\vspace{-0.1in}
\section{Game-Preserving Intervention and Gradient Routing}
\label{sec:scaling-routing}

Positive reward scaling separates strategic invariance from learning
stability.  We retain one diagnostic equation and one interpretive remark for
each theoretical result; complete statements, assumptions, and proofs appear
in Supplementary Material, Appendix~A.

\vspace{-0.1in}
\subsection{Strategic Invariance}

Let $J_i(\boldsymbol{\pi})$ be agent $i$'s discounted return under joint policy
$\boldsymbol{\pi}$.  Scaling every reward by $s>0$ gives the following formal
result.
\begin{theorem}[Strategic invariance]
\label{thm:main-strategic-invariance}
\begin{equation}
J_i^{(s)}(\boldsymbol{\pi})=sJ_i(\boldsymbol{\pi}).
\label{eq:main-strategic-invariance}
\end{equation}
\end{theorem}

\begin{remark}
Theorem~\ref{thm:main-strategic-invariance} preserves fixed-policy preferences,
best responses, Nash equilibria, and homogeneous linear incentive comparisons.
Reward scaling therefore changes the learning system without changing which
strategic solutions the game supports; it does not imply that a neural
optimizer is scale invariant.
\end{remark}

In \minex{}, let $P$, $R$, and $T$ be the punishment, cooperative, and one-shot
deviation returns, with $T>P$ and $\delta^\star=(T-R)/(T-P)$.
\begin{corollary}[Invariance of the repeated-game incentive threshold]
\label{cor:main-minex-threshold}
\begin{equation}
\delta^{\star,(s)}
=\frac{sT-sR}{sT-sP}
=\delta^\star .
\label{eq:main-minex-threshold}
\end{equation}
\end{corollary}

\noindent
Corollary~\ref{cor:main-minex-threshold} shows that positive reward scaling
preserves both profitability and incentive compatibility.  A scale-dependent
maintenance outcome therefore cannot be attributed to a different repeated
game or a different cooperation threshold.

\vspace{-0.1in}
\subsection{Matched Direct-Route Identification}

Let $\theta$ denote the shared actor torso, $\phi$ the value head,
$\mathcal{B}$ a fixed batch, $L_V$ the quadratic value loss,
$y_b^{(s)}$ the value target at scale $s$, and
$d_b^{(s)}=y_b^{(s)}-y_b^{(1)}$.  At a common parameter point, reward scaling
changes the value gradient by
\begin{lemma}[Value-gradient displacement]
\label{lem:main-value-displacement}
\begin{equation}
\nabla_\theta L_V^{(s)}-\nabla_\theta L_V^{(1)}
=-\frac{1}{|\mathcal{B}|}\sum_{b\in\mathcal{B}}
d_b^{(s)}\nabla_\theta V_{\theta,\phi}(o_b).
\label{eq:main-value-displacement}
\end{equation}
\end{lemma}

\noindent
Lemma~\ref{lem:main-value-displacement} isolates the scale-sensitive signal
available to a shared critic.  X0 sends this signal into the actor torso,
whereas X1 blocks it and X5 has no learned value objective.  The identity is
local to a common checkpoint and batch; it does not claim that later on-policy
trajectories remain matched.

Assume that the arms share the checkpoint, batch, and normalized policy
contribution, and let $\Delta\theta_{a,s}$ denote the first unclipped-SGD torso
update for arm $a$ at scale $s$.  Let $\mathcal{D}^{\theta}_s$ denote the
X0-minus-X1 difference between scales $s$ and $1$; $\alpha_{\mathrm{lr}}$
and $c_V$ are the learning rate and value-loss weight.  Then
\begin{theorem}[Direct-route scale dependence]
\label{thm:main-direct-route}
\begin{equation}
\mathcal{D}^{\theta}_s
=\frac{\alpha_{\mathrm{lr}}c_V}{|\mathcal{B}|}
\sum_{b\in\mathcal{B}}d_b^{(s)}
\nabla_\theta V_{\theta,\phi}(o_b).
\label{eq:main-direct-route}
\end{equation}
\end{theorem}

\begin{remark}
Theorem~\ref{thm:main-direct-route} shows that the first matched X0--X1
difference under unclipped SGD is exactly the direct value-to-torso effect,
subject to the matched-policy condition.  The statement identifies a
first-update mechanism, not the full later training trajectory.
\end{remark}

\begin{corollary}[Scale invariance of the stopped route under unclipped SGD]
\label{cor:main-x1-unclipped}
\begin{equation}
\Delta\theta_{\mathrm{X1},s}
=\Delta\theta_{\mathrm{X1},1}.
\label{eq:main-x1-unclipped}
\end{equation}
\end{corollary}

\noindent
Corollary~\ref{cor:main-x1-unclipped} states that X1 removes the direct
scale-sensitive torso term while retaining the critic and its forward
computation.  X5 also lacks this term, but its Monte-Carlo policy estimator
makes it a reference rather than a second exact control.

\vspace{-0.1in}
\subsection{Optimizer Coupling}

Stopping the direct route does not make the whole optimizer scale invariant.
For global clipping threshold $G$, let
$\chi_s=1\wedge G/\lVert g_{\mathrm{X1},s}\rVert_2$, where
$g_{\mathrm{X1},s}$ is the full X1 gradient and $g_{\pi,\theta}$ is its common
policy-to-torso component.  The resulting X1 torso contrast is
\begin{proposition}[Clipping-mediated scale dependence]
\label{prop:main-clipping-coupling}
\begin{equation}
\Delta\theta_{\mathrm{X1},s}^{\mathrm{clip}}
-\Delta\theta_{\mathrm{X1},1}^{\mathrm{clip}}
=-\alpha_{\mathrm{lr}}(\chi_s-\chi_1)g_{\pi,\theta}.
\label{eq:main-clipping-coupling}
\end{equation}
\end{proposition}

\noindent
Proposition~\ref{prop:main-clipping-coupling} explains why X1 can remain
scale-sensitive even with zero value-to-torso gradient: the value head changes
the global norm and therefore rescales the policy step whenever clipping is
active for at least one compared scale.  Adam can preserve or reshape scale
dependence already present in torso gradients through its coordinatewise
history, but it does not by itself transmit a detached value-head gradient into
the torso.

\vspace{-0.1in}
\subsection{From Displacement to Behavioral Risk}

Fix a common policy head $\bar\psi$ and let $h(\theta,\bar\psi)$ be a local
cooperation margin, with $h\geq0$ denoting success.  Assume that $h$ is twice
continuously differentiable and that its Hessian operator norm is bounded by
$L_h$ on the update neighborhood.  Let $\varepsilon_\Delta$ be the largest
norm among the four compared torso updates and $R_h$ the Taylor remainder.
Then
\begin{corollary}[First-order maintenance-margin effect]
\label{cor:main-margin-effect}
\begin{equation}
\begin{aligned}
\mathcal{D}^{h}_s
&:=\big[h(\theta+\Delta\theta_{\mathrm{X0},s},\bar\psi)
-h(\theta+\Delta\theta_{\mathrm{X0},1},\bar\psi)\big]\\
&\quad-\big[h(\theta+\Delta\theta_{\mathrm{X1},s},\bar\psi)
-h(\theta+\Delta\theta_{\mathrm{X1},1},\bar\psi)\big]\\
&=\frac{\alpha_{\mathrm{lr}}c_V}{|\mathcal{B}|}
\nabla_\theta h(\theta,\bar\psi)^\top
\sum_{b\in\mathcal{B}}d_b^{(s)}
\nabla_\theta V_{\theta,\phi}(o_b)+R_h,\\
|R_h|&\leq2L_h\varepsilon_\Delta^2 .
\end{aligned}
\label{eq:main-margin-effect}
\end{equation}
\end{corollary}

\noindent
Corollary~\ref{cor:main-margin-effect} separates route existence from
behavioral harm.  The direct route lowers a local cooperation margin only when
its displacement is adversely aligned beyond the approximation remainder; a
maintenance event additionally requires enough movement to cross the available
margin.  The result predicts a conditional local mechanism, not inevitable
failure from reward scaling or critic learning.

\vspace{-0.1in}
\subsection{Matched Learning Arms}

\paragraph{X0: direct value-to-torso route.}
The policy and value heads share the actor torso.  The combined loss updates
the torso, so value-to-torso gradients are generally nonzero.  Reward scaling
can therefore alter actor features through the shared value route even when
the matched policy-gradient contribution is unchanged.

\paragraph{X1: stop-gradient route control.}
The critic remains present and its head continues to learn, but the value
branch receives detached actor features.  Thus the value-head gradient is
nonzero while the direct value-to-torso gradient is exactly zero.  X0 versus
X1 isolates direct access under unclipped SGD; global clipping can retain an
indirect scale dependence through the critic's contribution to the full
gradient norm.

\paragraph{X5: no learned critic.}
X5 is an actor-only Monte-Carlo reference.  It removes both the direct route and
the learned value head.  Because its return estimator differs from X0/X1, X5
is a reference for critic-free optimization rather than a second exact
stop-gradient control.

The contrasts therefore answer different questions.  X0 versus X1 isolates
direct value-gradient access while retaining the critic and forward
architecture.  X1 versus X5 measures what remains when a critic is trained
without torso access, including differences in return estimation and, when
present, global clipping.  Failure common to X1 and X5 is not route-specific.

\begin{table}[t]
\centering
\scriptsize
\setlength{\tabcolsep}{3.5pt}
\begin{tabular}{lcccc}
\toprule
Arm & Critic & Value head & Value$\rightarrow$torso & Actor updates\\
\midrule
X0 & learned & yes & yes & yes\\
X1 & learned & yes & \textbf{no} & yes\\
X5 & none & n/a & n/a & yes\\
\bottomrule
\end{tabular}
\caption{\textbf{Matched routing arms.} X0 and X1 share the critic, forward
pass, and actor-update schedule, differing only in value-gradient access to
the shared torso; X5 is a critic-free reference.}
\label{tab:arms}
\end{table}

\vspace{-0.1in}
\section{Experiments}

\subsection{Experimental Settings}

\subsubsection{Environments and Metrics}
We evaluate cooperation maintenance in two structurally distinct multi-agent
environments, each with an environment-specific cooperation criterion and a
common temporal evaluation protocol:
\begin{itemize}
    \item \textbf{\minex{}.} A two-agent repeated production-and-exchange
    game with complementary resource advantages.  We use three profitable,
    incentive-compatible contexts spanning wide to narrow strategic margins
    (Table~\ref{tab:minexcontexts}); frozen evaluations assess specialization,
    exchange, resource flow, and return gain.

    \item \textbf{\cleanup{}.} A compact two-agent reconstruction of the
    CleanUp public-resource dilemma
    \citep{hughes2018inequity}, where cleaning sacrifices immediate harvest to
    preserve future resources.  A behavior-cloned scripted teacher provides
    the cooperative warm start; frozen evaluations require sustained return,
    apple collection, and pollution control.
\end{itemize}

Complete dynamics and operational success thresholds appear in Supplementary
Material Sections~2--3.

\begin{table}[t]
\centering
\small
\begin{tabular}{lrrrrr}
\toprule
Context & $H/L$ & $c$ & $q$ & $m_{\mathrm{IC}}$ & $m_{\mathrm{profit}}$\\
\midrule
H1 & 4.0 & 0.15 & 0.90 & 0.416 & 0.844\\
H2 & 2.5 & 0.40 & 0.90 & 0.191 & 0.178\\
H3 & 2.2 & 0.20 & 0.65 & 0.015 & 0.242\\
\bottomrule
\end{tabular}
\caption{\textbf{\minex{} strategic contexts.} All contexts support profitable,
incentive-compatible cooperation; H3 has the narrowest incentive margin.
$m_{\mathrm{IC}}$ and $m_{\mathrm{profit}}$ denote the incentive-compatibility
and profitability margins.}
\label{tab:minexcontexts}
\end{table}

All experimental cells use the frozen temporal endpoint defined above.
Evaluations use raw, unscaled rewards; survival and RMST summarize sustained
maintenance, while return and behavioral trajectories retain transient
failures and recoveries.

\subsubsection{Learning Configurations and Baselines}
We use the routing arms in Table~\ref{tab:arms}.  X0 versus X1 is the primary
matched contrast because it changes direct value-to-torso access while
retaining the critic and its forward computation.  X5 removes the learned
critic and uses a Monte-Carlo policy estimator, so it tests critic-free
optimization without serving as a second exact routing control.

The confirmatory \minex{} and boundary \cleanup{} cells use globally clipped
Adam; the opposite cells use unclipped SGD.  Adam retains clipping-mediated
coupling, whereas SGD provides the cleanest direct-route X0/X1 comparison.
Supplementary Material Sections~6 and 8.1 give the clipping analysis and
implementation details.

Positive reward scales are applied only during training. Because all scales
are strictly positive, they preserve fixed-policy strategic orderings while
changing value targets and optimizer dynamics. The X0/X1 comparison is the
primary routing contrast; the X5 arm tests whether observed failure also
arises without a learned critic.

\subsubsection{Training and Evaluation Protocol}
Every arm--scale pair within a seed loads byte-identical actor parameters and
uses matched initial evaluation episodes. X0 and X1 use the same critic
initialization and critic warm-up, during which actor parameters remain
frozen. Training targets use scaled rewards, whereas all reported evaluation
returns and cooperation thresholds use raw rewards. Reward scale therefore
cannot mechanically alter the evaluation criterion.

Training alternates on-policy collection, learning updates, and frozen
evaluation. Runs continue after confirmed events, and later training data
diverge with the policies despite matched initialization. Gradient audits use
the production update path. Supplementary Material Section~8.2 specifies the
step-0, randomness, checkpoint, and recovery controls.

\subsubsection{Local Behavioral-Margin Audit}
The local audit asks whether the measured value-to-torso route points toward a
behavioral failure boundary, rather than treating parameter movement itself as
harm. At frozen cooperative actors, we compare matched X0 and X1 production
updates at low and high reward scales. We then perturb only the shared torso
along the resulting route-by-scale and high-scale X0--X1 directions, keeping
the policy head fixed, and measure how cooperation margins change in both
directions. Repeating the audit across checkpoint or context variants,
training batches, and evaluation seeds separates a consistent local direction
from sampling noise. Because exact greedy behavior in \minex{} is locally
piecewise constant, its differentiable audit uses a cooperative soft-greedy
surrogate.

Table~\ref{tab:studies} summarizes the experimental cells. The confirmatory
cells test routing-specific scale effects in distinct environments and
optimizers, while the opposite optimizer--environment cells provide boundary
conditions. Exposure units are analyzed within cells rather than pooled.

\begin{table*}[t]
\centering
\small
\begin{tabular}{llclrrl}
\toprule
Environment & Optimizer & Clipping & Role & Runs & Horizon & Primary summary\\
\midrule
\minex{} & Adam & global norm 0.5 & confirmatory & 900 & 196,608 env. steps & hazard / RMST\\
\cleanup{} & SGD & none & confirmatory & 300 & 200 updates & paired RMST slope\\
\minex{} & SGD & none & boundary pilot & 225 & 196,608 env. steps & raw events / actor activity\\
\cleanup{} & Adam & global norm 0.5 & boundary & 300 & 200 updates & raw events / RMST\\
\bottomrule
\end{tabular}
\caption{\textbf{Experimental evidence matrix.} Confirmatory cells test
routing-specific scale effects across environments and optimizers; the
opposite pair probes all-arm survival or failure.}
\label{tab:studies}
\end{table*}

The confirmatory studies cross all three learning configurations with the
prespecified reward scales and fresh seeds. \minex{} uses Adam with global
clipping across all three strategic contexts, whereas \cleanup{} uses
unclipped SGD with paired seeds. Both studies evaluate the common initial
checkpoint and the full maintenance horizon.

\subsubsection{Statistical Analysis}
The statistical tests follow the intervention: reward scale should increase
maintenance risk more strongly when value gradients directly enter actor
features. For \minex{}, a discrete-time hazard model tests the X0 scale trend
and its contrasts with X1 and X5 while accounting for training time and
strategic context. For \cleanup{}, paired tests compare how each arm's RMST
changes with scale. We correct each hypothesis family for multiple testing.
When an arm has no observed events, we report survival and finite RMST
contrasts instead of forcing an unstable hazard estimate. The Supplementary
Material gives endpoint and environment details in Sections~2--3, audit
construction in Section~7, and implementation and statistical specifications
in Sections~8.1--8.3.

\vspace{-0.1in}
\subsection{Main Results}

The main experiments test whether reward scaling selectively destabilizes X0,
which exposes actor features to value gradients, relative to the matched X1
control and critic-free X5 reference.

\paragraph{\minex{} under Adam.}

\begin{figure}[t]
\centering
\includegraphics[width=0.99\columnwidth]{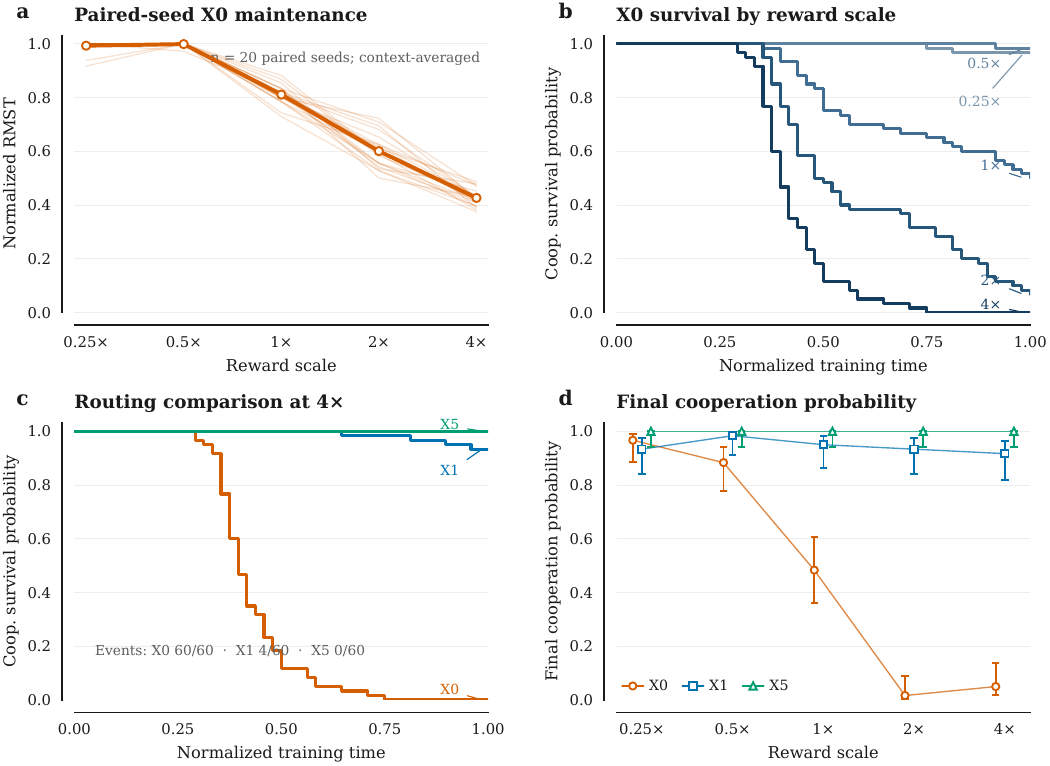}
\caption{\textbf{\minex{} / Adam confirmatory dynamics.}
\textbf{(a)} X0 RMST falls with scale across contexts.
\textbf{(b)} X0 failures become earlier and more frequent.
\textbf{(c)} At maximum scale, loss concentrates in X0; X1 approaches the
censoring ceiling and X5 has no confirmed events.
\textbf{(d)} Endpoint cooperation obscures failure timing.}
\label{fig:minexresults}
\end{figure}

Figure~\ref{fig:minexresults} shows a pronounced routing-by-scale interaction.
X0 remains cooperative at low scales but loses cooperation progressively as
scale increases. X1 stays near the censoring ceiling, and X5 has no confirmed
event. The hazard analysis confirms both the adverse X0 scale trend and its
difference from X1 after accounting for strategic context; descriptive X0
trends appear in all three contexts. Thus, under Adam, scale-sensitive
maintenance loss concentrates in the arm with the direct value-to-torso route.

\paragraph{\cleanup{} under unclipped SGD.}

\begin{figure}[t]
\centering
\includegraphics[width=0.99\columnwidth]{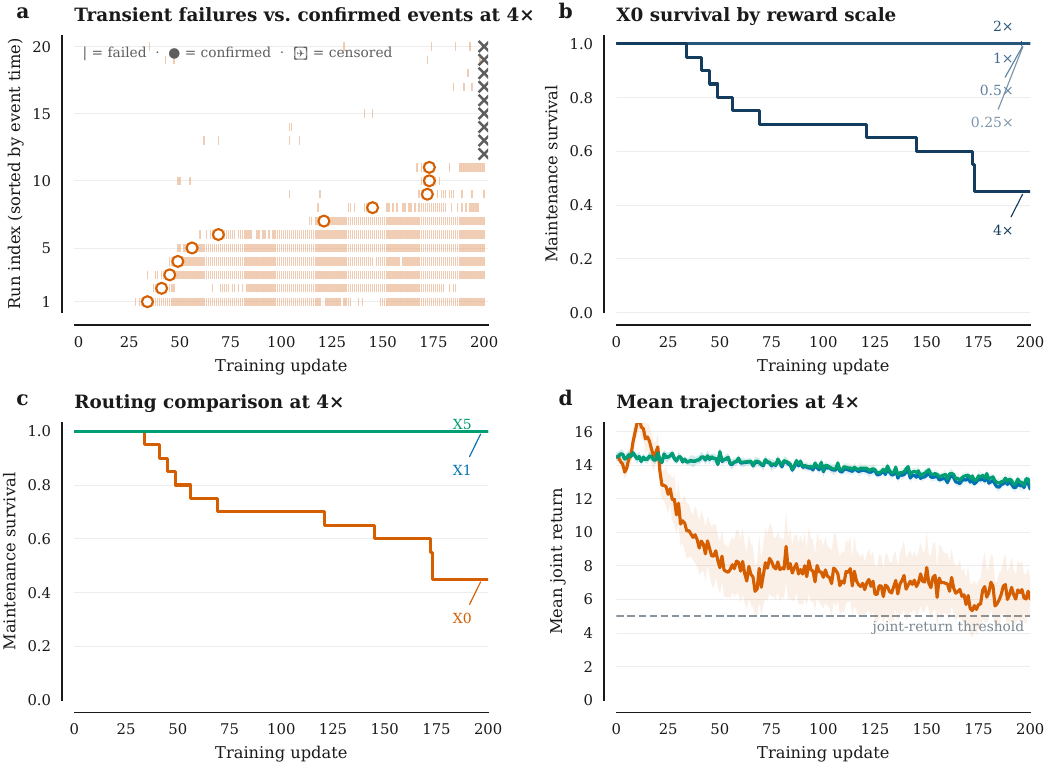}
\caption{\textbf{\cleanup{} / unclipped-SGD confirmatory dynamics.}
\textbf{(a)} The high-scale raster distinguishes transient failures,
confirmed events, and censoring.  \textbf{(b,c)} At maximum scale, survival
deteriorates and confirmed failures appear only in X0.
\textbf{(d)} Raw joint return reveals pre-event deterioration.}
\label{fig:cleanupresults}
\end{figure}

\cleanup{} produces a sharper high-scale boundary
(Figure~\ref{fig:cleanupresults}). X0 remains cooperative at lower scales but
develops sustained failures at the highest scale. Earlier isolated failures
sometimes recover, showing why a consecutive-failure endpoint is more
informative than first failure alone. X1 and X5 remain event-free across
scales and retain high raw joint return. Paired tests support the X0 scale
effect and both arm contrasts after multiple-testing correction. Within this
confirmatory \cleanup{}--SGD cell, the scale-sensitive loss is confined to X0,
where value gradients directly enter actor features.

\vspace{-0.1in}
\subsection{Mechanism and Robustness Diagnostics}

The outcome contrast identifies a routing-specific risk, but the mechanism
claim additionally requires evidence that the intended gradient route exists
and points in a behaviorally adverse direction.

\paragraph{Direct gradient audit and optimizer coupling.}

\begin{figure}[t]
\centering
\includegraphics[width=\columnwidth]{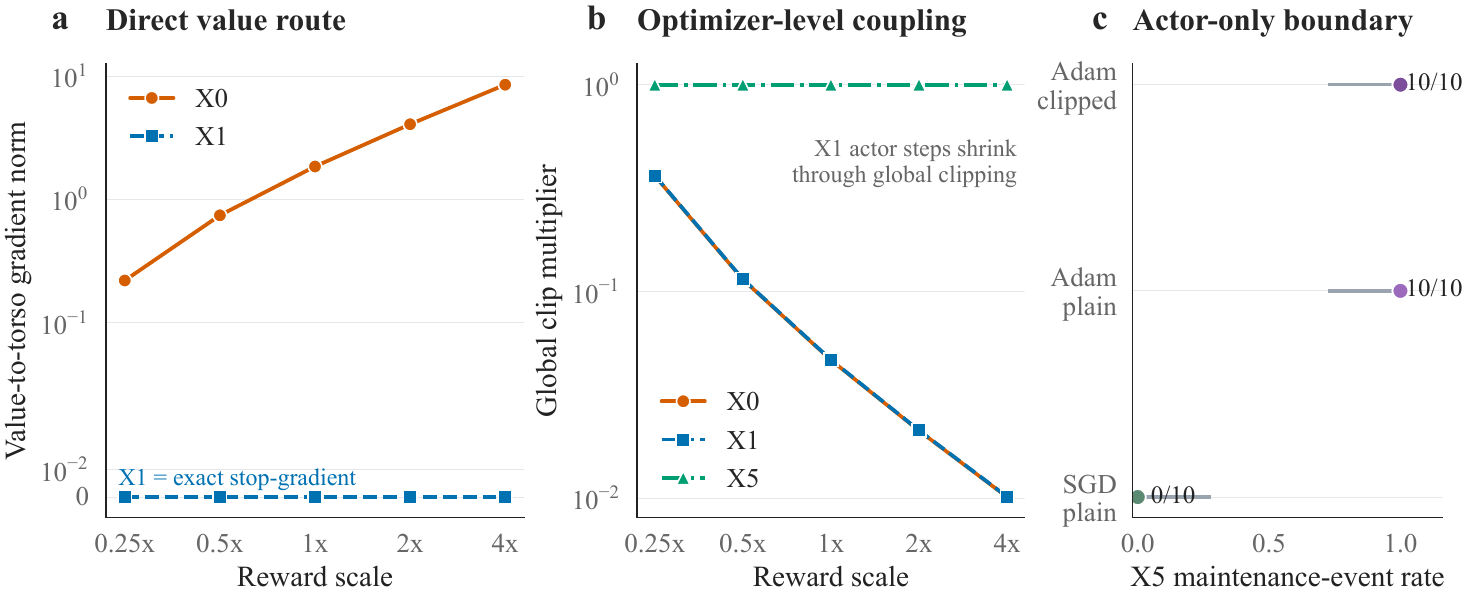}
\caption{\textbf{Mechanism diagnostics.}
\textbf{(a)} Direct value-to-torso gradients scale in X0 and vanish in X1.
\textbf{(b)} Under global clipping, critic gradients still shrink X1 actor
steps.  \textbf{(c)} X5 fails under Adam but not SGD, revealing
route-independent optimizer effects.}
\label{fig:mechanism}
\end{figure}

The production-path audit in Figure~\ref{fig:mechanism} verifies the treatment
directly: X0 has nonzero value-to-torso gradient, while X1 has exactly zero
gradient along that route despite an active value head.  Their actor parameters
diverge after one update.  The clipping audit also shows why stop-gradient is
not complete optimization independence: value-head gradients still enter the
global norm and can rescale X1's policy update.  The unclipped \cleanup{} study
removes this indirect route, yielding the cleanest X0/X1 comparison.

\paragraph{Behavioral-margin alignment.}
The local audit supports adverse alignment in \cleanup{} but yields only
partial evidence in \minex{} under Adam. In \cleanup{}, both the
route-by-scale direction and the high-scale X0--X1 direction lower every
active greedy cooperation margin. Perturbing the full actor produces nearly
the same response as perturbing only the torso, which localizes the adverse
effect to the shared representation. In \minex{}, exact greedy behavior is
locally flat. Its cooperative soft-greedy surrogate shows a partial,
contrast-dependent result: the high-scale X0--X1 direction lowers the flow
margin, but evidence for return and specialization is inconclusive, and the
Adam route-by-scale direction is locally null. These local directions diagnose
risk; they do not by themselves establish that a perturbation crosses the
failure threshold. Full estimates, intervals, and linearity checks appear in
the Supplementary Material.

\paragraph{Local cooperative-checkpoint neighborhood.}

\begin{figure}[t]
\centering
\includegraphics[width=\columnwidth]{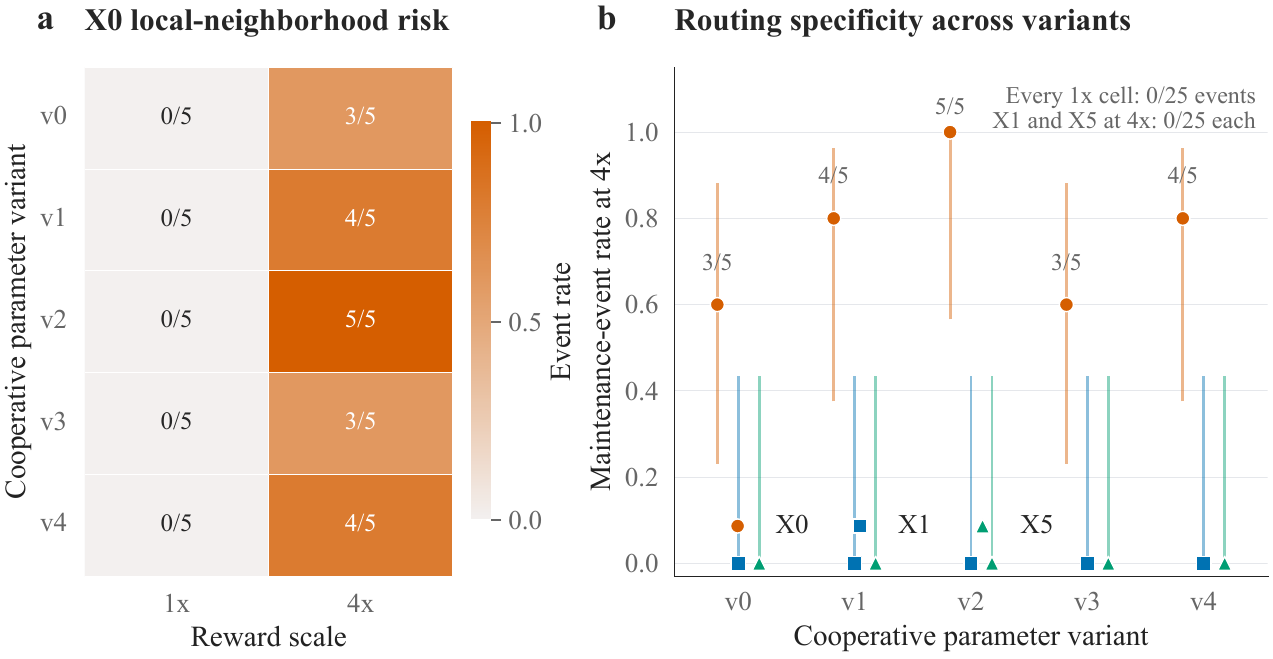}
\caption{\textbf{Local warm-start robustness in \cleanup{} / SGD.}
\textbf{(a)} Across the base checkpoint and four local cooperative
perturbations, X0 fails frequently only at high scale.
\textbf{(b)} High-scale failures remain exclusive to X0; X1 and X5 remain
event-free.  Variants are perturbations, not independently trained policies.}
\label{fig:variants}
\end{figure}

The main studies use one actor to maximize causal matching. To test whether
the result depends on that exact checkpoint, we repeat the comparison at
four nearby cooperative parameter perturbations. At high scale, maintenance
loss again concentrates in X0, whereas X1 and X5 remain event-free
(Figure~\ref{fig:variants}); at low scale, all arms remain event-free over the
tested horizon. The result is therefore robust within a local cooperative
neighborhood, although these
perturbations are not independently trained solutions.

\vspace{-0.1in}
\subsection{Boundary Conditions}

The final experiments test whether direct value routing is universally
sufficient or necessary for maintenance failure.

\begin{figure}[t]
\centering
\includegraphics[width=\columnwidth]{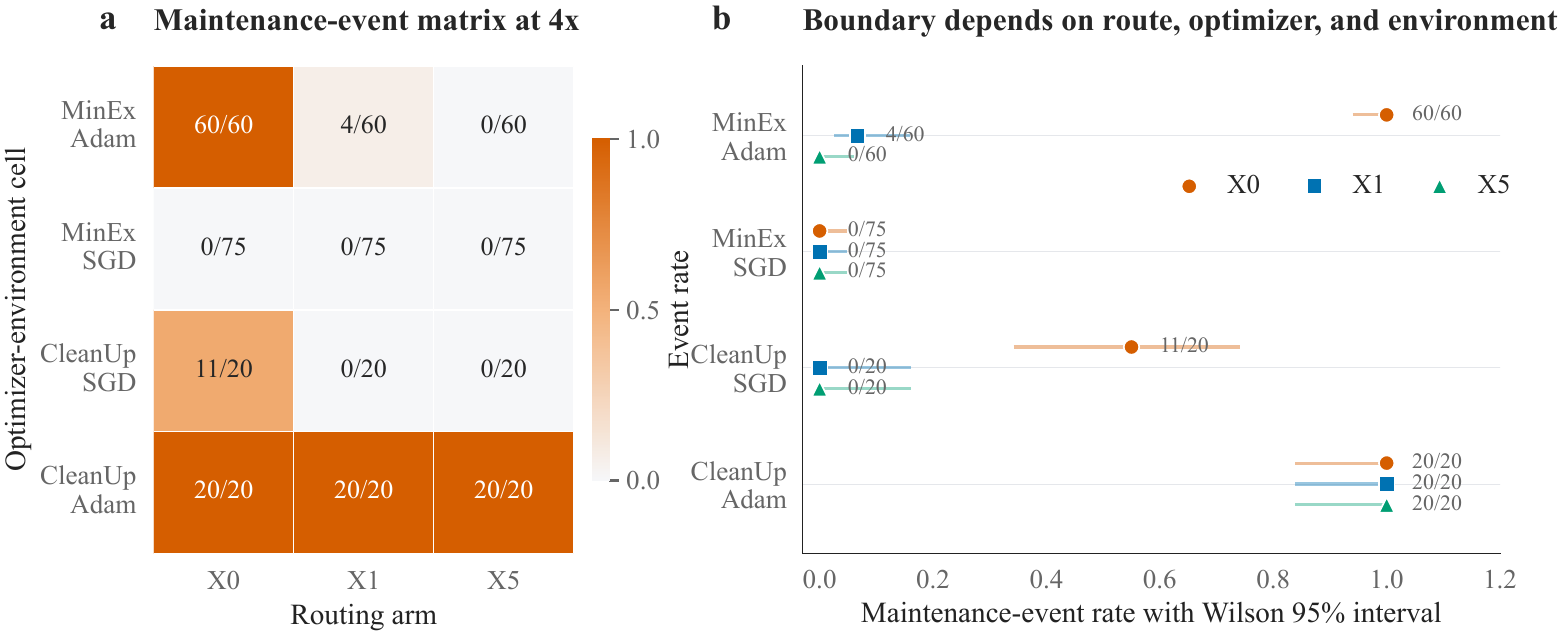}
\caption{\textbf{Optimizer--environment boundary at reward scale 4.}
\textbf{(a)} Event rates show routing-specific failure, all-arm survival in
\minex{}--SGD, and all-arm failure in \cleanup{}--Adam.
\textbf{(b)} Wilson intervals confirm the pattern; direct value routing is
neither necessary nor sufficient within this matrix.}
\label{fig:matrix}
\end{figure}

The boundary matrix in Figure~\ref{fig:matrix} prevents a universal causal
interpretation.  With unclipped SGD, every \minex{} run remains event-free
despite continuing parameter updates.  Conversely, \cleanup{} with Adam fails
in every arm, including X5, so neither a learned critic nor direct value
routing is necessary in that cell.  The routing effect is therefore
conditional on both optimizer and environment across the tested matrix.

\vspace{-0.1in}
\section{Discussion and Conclusion}

\noindent\textbf{Interpretation, limitations, and conclusion.}
The matched design verifies direct value-to-torso routing, and confirmatory
contrasts show that scale-sensitive maintenance loss concentrates in X0.  The
local audit supplies the behavioral condition: adverse alignment is clear in
\cleanup{}--SGD but weaker and contrast-dependent in \minex{}--Adam.  Route
existence, adverse alignment, sufficient movement, and a confirmed event are
therefore distinct evidentiary steps.  Positive scaling preserves strategic
comparisons but changes maintenance time, so game stability does not certify
learning stability.  Our evidence is limited to two compact environments,
including a reconstructed \cleanup{}, and to verified warm starts rather than
cooperation discovery.  Checkpoint perturbations are local diagnostics, not
independently trained solutions: the audits estimate derivatives rather than
threshold crossings; the \minex{} surrogate has curvature that limits
extrapolation; and gradient and clipping audits do not fully capture later
on-policy trajectories.  Operational event definitions, finite-horizon
survival, and complete separation provide only finite descriptive evidence.
Within these limits, continued optimization can erase verified cooperation
without changing strategic ordering when value-gradient displacement is
adversely aligned and reaches a failure boundary.  Direct value-gradient
access is thus a conditional, scale-sensitive risk shaped by the optimizer
and environment.  MARL evaluations should report acquisition and survival
trajectories and treat routing, reward scale, clipping, and optimizer as
algorithmic choices.

\bibliography{references}

\clearpage
\setcounter{section}{0}
\setcounter{secnumdepth}{2}
\captionsetup{font=small,labelfont=bf}
\setlength{\columnsep}{0.25in}
\setlist[itemize]{nosep,leftmargin=1.4em}
\setlist[enumerate]{nosep,leftmargin=1.7em}
\setcounter{secnumdepth}{2}
\onecolumn
\raggedbottom
\raggedcolumns
\begin{center}
{\LARGE\bfseries Supplementary Material\par}
\vspace{6pt}
{\Large\bfseries After Cooperation Is Learned: Gradient Routing and\\
Optimizer-Dependent Maintenance in Multi-Agent Reinforcement Learning\par}
\vspace{12pt}
{\large Chaoyuan Hao, Wentao Yue, Tianyou Lai,\\
Hongji Li, Jiayi Zhou, Qingyu Mao, Qilei Li\par}
\end{center}
\vspace{14pt}

\begin{multicols}{2}
\section{Appendix Overview}

This appendix records the operational definitions, additional numerical
results, implementation controls, and artifact provenance needed to interpret
the main paper.  It intentionally distinguishes confirmatory studies from
pilots and local robustness diagnostics.  Legacy \cleanup{} CL-8 artifacts are
excluded from every table and figure; the paper uses the recovered and
re-executed CL-9 studies only.

\section{Maintenance Endpoint in Detail}
\label{sec:supp-endpoint}

For run $r$ and evaluation checkpoint $k$, let
$C_{rk}\in\{0,1\}$ denote environment-specific success after averaging the
frozen evaluation episodes.  True step 0 must satisfy $C_{r0}=1$.  Define the
first failure checkpoint
\begin{equation}
F_r=\min\{k>0:C_{rk}=0\}.
\end{equation}
For confirmation width $w=5$, the maintenance-event checkpoint is
\begin{equation}
E_r=\min\{k\geq w:C_{r,k-w+1}=\cdots=C_{rk}=0\}.
\end{equation}
If no such checkpoint occurs before horizon $\tau$, the run is right-censored.
RMST is computed from the Kaplan--Meier survival function through the frozen
horizon.  Joint-return AUC is integrated against the actual training-time
coordinate, not checkpoint index.

\paragraph{\minex{} success.}
The deterministic evaluation conjunction is:
\begin{enumerate}
    \item specialization score $\geq 0.80$;
    \item specialization directions match comparative advantage;
    \item balanced exchange $\geq 0.75$;
    \item resource flow $\geq 0.60$ of cooperative-oracle flow;
    \item normalized return gain over autarky $\geq 0.60$.
\end{enumerate}

\paragraph{\cleanup{} success.}
Each checkpoint averages 64 raw-reward episodes and succeeds iff
\begin{equation}
\overline{J}_{\mathrm{joint}}\geq5,\qquad
\overline{P}_{\mathrm{pollution}}\leq3.16,\qquad
\overline{N}_{\mathrm{apples}}\geq5.
\end{equation}
The threshold 3.16 is the frozen median from the CL-2F calibration.  An older
YAML value of 3.59 is deprecated and is not used in CL-9.

\section{Environment Details}
\label{sec:supp-environments}

\paragraph{\minex{} dynamics.}
Agent 1 has productivity $(H,L)$ and agent 2 has $(L,H)$, where $H>L>0$.
Each round, agents select a production allocation
$x_i\in\{0,0.25,0.5,0.75,1\}$ and transfer fractions for both goods in
$\{0,0.5,1\}$.  A transferred unit loses fraction $c$, and post-transfer
consumption $(a_i,b_i)$ gives utility $u_i=\sqrt{a_i b_i}$.  Autarkic
production yields $P=\sqrt{HL}/2$, while symmetric full specialization with
half of the advantage good transferred yields $R=H\sqrt{1-c}/2$.  If $T$ is
the exact one-shot best-deviation payoff and $\delta=\gamma q$ is the
effective relationship discount, grim-trigger cooperation is incentive
compatible when
\begin{equation}
\delta\geq\delta^\star=\frac{T-R}{T-P}.
\end{equation}

\paragraph{\cleanup{} dynamics.}
The reconstruction uses a two-agent $5\times4$ grid and a 64-step horizon.
Agents can move, harvest apples, clean river pollution, or zap.  Dirt appears
with probability 0.35, zaps succeed with probability 0.9, and a successful zap
freezes its target for five steps.  A scripted cooperative teacher and
behavior cloning produce the warm start, which succeeds in 99.4\% of 512
held-out frozen evaluations.
\end{multicols}

\section{Additional \minex{} Results}

\subsection{Full Arm--Scale Survival Table}

\begin{table}[H]
\centering
\small
\setlength{\tabcolsep}{3.5pt}
\caption{\minex{} / Adam confirmatory survival summary.  Higher reward scales
increase X0 event rates and shorten its maintenance, whereas X1 remains near
the censoring ceiling and X5 has no events.  RMST is normalized by the
196,608-step horizon.}
\label{tab:app-minexfull}
\begin{tabular}{lrrrrrrr}
\toprule
Arm & Scale & Events & Censored & Event rate & RMST & First-failure rate & Final cooperation\\
\midrule
X0 & 0.25 & 2 & 58 & 0.033 & 0.9927 & 0.033 & 0.967\\
X0 & 0.50 & 1 & 59 & 0.017 & 0.9986 & 0.083 & 0.883\\
X0 & 1.0 & 30 & 30 & 0.500 & 0.8108 & 0.533 & 0.483\\
X0 & 2.0 & 56 & 4 & 0.933 & 0.6003 & 0.967 & 0.017\\
X0 & 4.0 & 60 & 0 & 1.000 & 0.4264 & 1.000 & 0.050\\
\midrule
X1 & 0.25 & 3 & 57 & 0.050 & 0.9958 & 0.050 & 0.933\\
X1 & 0.50 & 1 & 59 & 0.017 & 0.9938 & 0.017 & 0.983\\
X1 & 1.0 & 2 & 58 & 0.033 & 0.9903 & 0.050 & 0.950\\
X1 & 2.0 & 4 & 56 & 0.067 & 0.9913 & 0.067 & 0.933\\
X1 & 4.0 & 4 & 56 & 0.067 & 0.9885 & 0.083 & 0.917\\
\midrule
X5 & 0.25 & 0 & 60 & 0.000 & 1.0000 & 0.000 & 1.000\\
X5 & 0.50 & 0 & 60 & 0.000 & 1.0000 & 0.000 & 1.000\\
X5 & 1.0 & 0 & 60 & 0.000 & 1.0000 & 0.000 & 1.000\\
X5 & 2.0 & 0 & 60 & 0.000 & 1.0000 & 0.000 & 1.000\\
X5 & 4.0 & 0 & 60 & 0.000 & 1.0000 & 0.000 & 1.000\\
\bottomrule
\end{tabular}
\end{table}

\begin{multicols}{2}
\subsection{Context Effects}

\begin{table}[H]
\centering
\small
\caption{X0 scale sensitivity appears in all three strategic contexts despite
their different incentive margins.  Each context pools five scales and 20
seeds per scale.}
\label{tab:app-context}
\begin{tabular}{lrrrr}
\toprule
Context & Events/100 & RMST & Scale slope & SE\\
\midrule
H1 & 60/100 & 0.665 & 1.935 & 0.191\\
H2 & 44/100 & 0.767 & 2.776 & 0.311\\
H3 & 45/100 & 0.865 & 1.957 & 0.258\\
\bottomrule
\end{tabular}
\end{table}

The H3 incentive margin is only 0.0147, yet its X0 scale slope is close to that
of H1, whose incentive margin is 0.416.  This does not imply incentive margin
is irrelevant; it shows that the confirmatory scale effect is not attributable
only to one context.
\end{multicols}

\section{Additional \cleanup{} Results}

\subsection{Full Confirmatory Table}

\begin{table}[H]
\centering
\small
\setlength{\tabcolsep}{3.5pt}
\caption{\cleanup{} / unclipped-SGD confirmatory summary.  Confirmed failures
occur only for high-scale X0; X1 and the critic-free X5 reference remain
event-free across scales.}
\label{tab:app-cleanupfull}
\begin{tabular}{lrrrrrrr}
\toprule
Arm & Scale & Events & Event rate & Mean RMST & Joint AUC & First-failure rate & Recovery rate\\
\midrule
X0 & 0.25 & 0/20 & 0.00 & 200.0 & 13.764 & 0.00 & 0.00\\
X0 & 0.50 & 0/20 & 0.00 & 200.0 & 13.751 & 0.00 & 0.00\\
X0 & 1.0 & 0/20 & 0.00 & 200.0 & 13.651 & 0.00 & 0.00\\
X0 & 2.0 & 0/20 & 0.00 & 200.0 & 13.672 & 0.00 & 0.00\\
X0 & 4.0 & 11/20 & 0.55 & 143.9 & 8.389 & 0.90 & 0.40\\
\midrule
X1 & 0.25 & 0/20 & 0.00 & 200.0 & 13.770 & 0.00 & 0.00\\
X1 & 0.50 & 0/20 & 0.00 & 200.0 & 13.771 & 0.00 & 0.00\\
X1 & 1.0 & 0/20 & 0.00 & 200.0 & 13.772 & 0.00 & 0.00\\
X1 & 2.0 & 0/20 & 0.00 & 200.0 & 13.773 & 0.00 & 0.00\\
X1 & 4.0 & 0/20 & 0.00 & 200.0 & 13.776 & 0.00 & 0.00\\
\midrule
X5 & 0.25 & 0/20 & 0.00 & 200.0 & 13.877 & 0.00 & 0.00\\
X5 & 0.50 & 0/20 & 0.00 & 200.0 & 13.877 & 0.00 & 0.00\\
X5 & 1.0 & 0/20 & 0.00 & 200.0 & 13.877 & 0.00 & 0.00\\
X5 & 2.0 & 0/20 & 0.00 & 200.0 & 13.877 & 0.00 & 0.00\\
X5 & 4.0 & 0/20 & 0.00 & 200.0 & 13.877 & 0.00 & 0.00\\
\bottomrule
\end{tabular}
\end{table}

\begin{multicols}{2}
\subsection{Local Checkpoint-Neighborhood Robustness}

\begin{table}[H]
\centering
\small
\caption{Across the base checkpoint and four local variants, high-scale
failures remain confined to X0; X1 and X5 have no events.  Each cell uses five
seeds.}
\label{tab:app-variants}
\begin{tabular}{lrrr}
\toprule
Checkpoint & X0 events & X1 events & X5 events\\
\midrule
Base & 3/5 & 0/5 & 0/5\\
Perturbation 1 & 4/5 & 0/5 & 0/5\\
Perturbation 2 & 5/5 & 0/5 & 0/5\\
Perturbation 3 & 3/5 & 0/5 & 0/5\\
Perturbation 4 & 4/5 & 0/5 & 0/5\\
\midrule
Total & 19/25 & 0/25 & 0/25\\
\bottomrule
\end{tabular}
\end{table}

At scale 1, all 75 variant runs are event-free.  The perturbations are local
changes around the cooperative checkpoint and should not be described as
independently trained checkpoints.
\end{multicols}

\subsection{Adam Boundary}

\begin{table}[H]
\centering
\scriptsize
\caption{\cleanup{} / Adam boundary: all cells have 20/20 events.  RMST and
joint AUC show that identical event rates do not imply identical dynamics.}
\label{tab:app-adam}
\begin{tabular}{lrrrrrrrrrr}
\toprule
& \multicolumn{5}{c}{Mean RMST by scale} & \multicolumn{5}{c}{Mean joint AUC by scale}\\
\cmidrule(lr){2-6}\cmidrule(lr){7-11}
Arm & 0.25 & 0.5 & 1 & 2 & 4 & 0.25 & 0.5 & 1 & 2 & 4\\
\midrule
X0 & 32.4 & 32.0 & 50.2 & 59.9 & 30.9 & 7.01 & 6.38 & 5.63 & 4.91 & 4.03\\
X1 & 26.4 & 25.3 & 21.4 & 20.8 & 21.0 & 7.39 & 7.01 & 6.96 & 6.20 & 5.57\\
X5 & 21.4 & 21.4 & 21.4 & 21.4 & 21.4 & 4.79 & 4.82 & 4.82 & 4.82 & 4.82\\
\bottomrule
\end{tabular}
\end{table}

\section{Mechanism and Clipping Audit}
\label{sec:supp-mechanism}

\begin{table}[H]
\centering
\small
\caption{Production-path gradient audit.  The direct value-to-torso gradient
is nonzero and grows with scale in X0 but is exactly zero in X1.  Norms describe
the audited update and are not averaged treatment effects.}
\label{tab:app-grad}
\begin{tabular}{lrrrrr}
\toprule
Arm & Scale & Policy$\rightarrow$torso & Value$\rightarrow$torso & Value$\rightarrow$head & Global clip multiplier\\
\midrule
X0 & 0.25 & 0.158 & 0.220 & 2.726 & 0.3625\\
X0 & 0.50 & 0.155 & 0.743 & 8.705 & 0.1143\\
X0 & 1.0 & 0.153 & 1.850 & 21.359 & 0.0466\\
X0 & 2.0 & 0.152 & 4.069 & 46.913 & 0.0212\\
X0 & 4.0 & 0.151 & 8.510 & 98.135 & 0.0102\\
\midrule
X1 & 0.25 & 0.158 & 0.000 & 2.726 & 0.3638\\
X1 & 0.50 & 0.155 & 0.000 & 8.705 & 0.1148\\
X1 & 1.0 & 0.153 & 0.000 & 21.359 & 0.0468\\
X1 & 2.0 & 0.152 & 0.000 & 46.913 & 0.0213\\
X1 & 4.0 & 0.151 & 0.000 & 98.135 & 0.0102\\
\midrule
X5 & all & 0.150 & n/a & n/a & 1.0000\\
\bottomrule
\end{tabular}
\end{table}

\begin{multicols}{2}
This table distinguishes direct and indirect coupling.  X1 eliminates
value-to-torso gradients exactly.  Under a global norm, however, its value-head
gradient still determines the clip multiplier applied to policy gradients.
The unclipped-SGD \cleanup{} confirmation removes this specific coupling.

\section{Local Behavioral-Margin Audit}
\label{sec:supp-local-audit}

At frozen cooperative actors, we execute one matched production-path update
for X0 and X1 at reward scales 1 and 4.  We retain the torso coordinates of
the arm-by-scale difference-in-differences and the scale-4 X0--X1 contrast,
keeping the checkpoint policy head fixed.  For each direction $D$, we evaluate
the cooperation margins along
\begin{equation}
\epsilon\in\{-2,-1,-0.5,0,0.5,1,2\}.
\end{equation}
The centered difference between $\epsilon=0.5$ and $\epsilon=-0.5$ estimates
local alignment; centered curvature and linear remainder diagnose departures
from local linearity.  The \cleanup{} audit uses five checkpoint variants,
five matched batch seeds, and three evaluation seeds with 512 episodes each.
The \minex{} audit crosses three contexts, five batch seeds, and three
evaluation seeds with 256 episodes each.  We first average checkpoint- or
context-level fits within each batch seed, then form a two-sided $t_4$
interval over the five batch clusters.  The \minex{} soft-greedy surrogate
uses temperature 0.5.

\begin{table}[H]
\centering
\scriptsize
\caption{Centered local behavioral-margin estimates.  Both \cleanup{}
directions point toward its cooperation boundary; in \minex{}, only the flow
margin under the high-scale X0--X1 direction is clearly adverse.  Brackets
give 95\% batch-cluster intervals.}
\label{tab:app-margin-audit}
\resizebox{\columnwidth}{!}{%
\begin{tabular}{lll}
\toprule
Environment / direction & Margin & Estimate [95\% CI]\\
\midrule
\cleanup{} arm-by-scale & Return / apples &
$-0.0628$ [$-0.1177,-0.0078$]\\
\cleanup{} arm-by-scale & Pollution &
$-0.00526$ [$-0.00736,-0.00316$]\\
\cleanup{} scale-4 X0--X1 & Return / apples &
$-0.0644$ [$-0.1003,-0.0284$]\\
\cleanup{} scale-4 X0--X1 & Pollution &
$-0.00554$ [$-0.00678,-0.00431$]\\
\minex{} scale-4 X0--X1 & Flow &
$-6.69{\times}10^{-4}$ [$-1.26{\times}10^{-3},-7.38{\times}10^{-5}$]\\
\bottomrule
\end{tabular}
}
\end{table}

Full-actor \cleanup{} perturbations are nearly identical to torso-only
perturbations.  Its central-neighborhood linear remainder stays below 0.7\%
of the baseline margin.  In \minex{}, return and specialization move
downward under the scale-4 X0--X1 direction, but their conservative intervals
touch or cross zero; the Adam arm-by-scale direction is locally null.  The
\minex{} linear remainder stays below 0.8\% of the baseline margin.

\section{Implementation and Reproducibility Details}
\label{sec:supp-implementation}

\subsection{Architecture and Updates}
\label{sec:supp-architecture}

Both environments use parameter sharing across the two agents.  The X0 and X1
actors have identical forward computations; X1 detaches the feature tensor
passed to the value branch.  X0 and X1 critics are initialized identically per
seed.  X5 omits value parameters and computes complete-episode Monte-Carlo
returns.  Actor parameters are trainable in all arms.

\begin{table}[H]
\centering
\scriptsize
\caption{Matched implementation details for the confirmatory cells.  The
\minex{} cell uses clipped Adam and the \cleanup{} cell uses unclipped SGD;
within each cell, X0 and X1 otherwise share the actor, critic computation, and
update schedule.}
\label{tab:app-implementation}
\resizebox{\columnwidth}{!}{%
\begin{tabular}{lll}
\toprule
Component & \minex{}--Adam & \cleanup{}--SGD\\
\midrule
Observation / actions & 14 / factorized $5,3,3$ & 82 / categorical 6\\
Shared actor torso & Dense(128)--tanh--Dense(128)--tanh &
Dense(128)--tanh--Dense(128)--tanh\\
Policy / value heads & three linear / scalar & six-way linear / scalar\\
Parallel environments & 16 (X0/X1); 8 (X5) & 8\\
Rollout length & 128 (X0/X1); 512 (X5) & 64\\
Epochs / minibatches & 4 / 4 & 4 / 4\\
Optimizer / learning rate & Adam / $3{\times}10^{-4}$ & SGD / $10^{-3}$\\
Discount / GAE $\lambda$ & 0.99 / 0.95 & 1.0 / 1.0\\
PPO clip / value / entropy & 0.2 / 0.5 / 0.01 & 0.2 / 0.5 / 0\\
Global norm clipping & 0.5 & none\\
Frozen evaluation & 256 episodes every 4,096 steps &
64 episodes every update\\
\bottomrule
\end{tabular}
}
\end{table}

For X0, the production loss has schematic form
\begin{equation}
L_{\mathrm{X0}}=L_{\pi}+c_VL_V-c_H\mathcal{H},
\end{equation}
and a single backward pass updates the shared torso and both heads.  For X1,
$L_V$ receives $\mathrm{stopgrad}(h_\theta)$, which preserves value-head
learning but gives zero derivative with respect to $\theta$.  For X5,
$L_V$ is absent.

\subsection{True-Step-0 and Randomness Controls}
\label{sec:supp-randomness}

\begin{enumerate}
    \item Load the common warm-start actor.
    \item For X0/X1, complete matched critic warm-up without changing actor
    parameters.
    \item Evaluate true step 0 before any maintenance-training actor update.
    \item Use identical evaluation episode keys across arms and scales within
    seed.
    \item Continue training, then run frozen evaluation at every preregistered
    checkpoint.
\end{enumerate}
The evaluation path uses raw rewards and does not advance training RNG.
Saved maintenance fields are independently recomputed from checkpoint rows.
Runs continue after confirmed events so that recovery and raw-return
trajectories remain observable, and gradient-route audits invoke the same
production update path used in training.
The confirmatory maintenance seeds are 140--159 for \minex{} and 400--419 for
\cleanup{}.  Conditions sharing a seed use matched initial actors and
evaluation keys; later on-policy data are not forced to match after policies
diverge.

\subsection{Confirmatory Statistical Analysis}
\label{sec:supp-statistics}

For \minex{}, each at-risk checkpoint contributes one row to a discrete-time
logistic hazard model containing log time, reward scale, arm, arm--scale
interactions, and strategic context.  The confirmatory hypotheses test whether
the X0 scale slope is positive and whether it exceeds the X1 and X5 slopes.
Repeated checkpoints and matched conditions sharing a maintenance seed form
one covariance cluster.  The primary model therefore estimates a
routing-specific scale effect while accounting for context rather than pooling
contexts as independent studies.

For \cleanup{}, each seed yields an RMST slope across reward scales.  Paired
tests ask whether the X0 slope is negative and more negative than the
corresponding X1 and X5 slopes.  Each hypothesis family uses Holm correction.
The adjusted tests support the X0 scale effect and both arm contrasts
($p=0.002$).  When an arm has no observed events, survival, event counts, and
finite RMST contrasts replace a non-finite maximum-likelihood hazard
coefficient.

For \minex{}, the clustered hazard analysis supports both the positive X0
scale slope and its contrast with X1.  X5 has no confirmed events, so its
separation from X0 is reported descriptively rather than through a truncated
coefficient.  Context-stratified summaries and complete arm--scale survival
tables appear above.

\begin{table}[H]
\centering
\small
\caption{Seed-clustered \minex{} hazard contrasts.  The X0 hazard increases
with reward scale and more steeply than the matched X1 hazard.}
\label{tab:app-cluster-hazard}
\begin{tabular}{lrrr}
\toprule
Contrast & Estimate & Cluster SE & Holm $p$\\
\midrule
X0 scale slope $>0$ & 1.993 & 0.087 & $<10^{-100}$\\
X0 slope $>$ X1 slope & 1.807 & 0.270 & $1.04{\times}10^{-11}$\\
\bottomrule
\end{tabular}
\end{table}

\subsection{Artifact Authority}

The paper-facing data directory contains:
{\small
\begin{itemize}
    \item \path{minex/dv_r1_final_survival_summary.csv} and associated
    confirmatory model tables;
    \item \path{cleanup_cl9/}\allowbreak\path{run_summary.csv} and
    \path{cleanup_cl9/}\allowbreak\path{checkpoint_metrics.parquet} for the
    final SGD confirmation;
    \item \path{cleanup_cl9/}\allowbreak\path{adam_run_summary.csv} for the Adam
    boundary;
    \item \path{cleanup_cl9/}\allowbreak\path{variant_*} for the local
    checkpoint study;
    \item \path{cleanup_cl9/}\allowbreak\path{mechanism_audit.json} and
    \path{cleanup_cl9/}\allowbreak\path{clipping_coupling.json} for route
    diagnostics;
    \item \path{cross_environment/}\allowbreak
    \path{optimizer_environment_boundary_scale4.csv} for the four-cell
    summary.
\end{itemize}
}
Every paper-facing file has a SHA-256 entry in
\path{paper_data_manifest.csv}.  CL-8 is scientifically invalidated and
excluded.

\section{Claim Boundary}

\paragraph{Supported.}
\begin{itemize}
    \item Cooperation discovery and cooperation maintenance are empirically
    distinct evaluation problems.
    \item Positive reward scaling can preserve the strategic game while
    changing neural maintenance dynamics.
    \item Matched outcome contrasts associate direct value-to-torso routing
    with scale-sensitive maintenance loss in \minex{}--Adam and
    \cleanup{}--SGD.
    \item Local adverse alignment is clear in \cleanup{}--SGD but partial and
    contrast-dependent in \minex{}--Adam.
    \item Maintenance outcomes vary across optimizer--environment cells.
    \item X1 blocks the audited direct route while retaining value-head
    learning.
\end{itemize}

\paragraph{Not supported.}
\begin{itemize}
    \item Critics always destroy cooperation.
    \item Direct value routing is the unique or universally necessary cause of
    maintenance failure.
    \item Stop-gradient guarantees cooperation.
    \item \cleanup{} is an official benchmark implementation.
    \item The local perturbation study is replication across independently
    learned cooperative solutions.
    \item A censored run is permanently stable beyond the tested horizon.
\end{itemize}
\end{multicols}

\clearpage
\setcounter{secnumdepth}{1}
\appendix
\section{Formal Analysis and Proofs}
\label{app:theory-proofs}

This appendix gives the formal statements and proofs underlying the
conclusion-level account in the main paper.

\subsection{Strategic Invariance of Positive Reward Scaling}

For agent $i$, let
\begin{equation}
J_i(\boldsymbol{\pi})
:=\mathbb{E}_{\boldsymbol{\pi}}
\left[\sum_{t=0}^{\infty}\gamma^t r_{i,t}\right]
\end{equation}
be its discounted environment return under joint policy
$\boldsymbol{\pi}$.  For a scale $s>0$, define
$r_{i,t}^{(s)}:=s r_{i,t}$ and denote the corresponding return by
$J_i^{(s)}$.

\begin{suppTheorem}[Strategic invariance]
\label{prop:strategic-invariance}
For every fixed joint policy $\boldsymbol{\pi}$ and every $s>0$,
\begin{equation}
J_i^{(s)}(\boldsymbol{\pi})=sJ_i(\boldsymbol{\pi}).
\label{eq:return-scaling}
\end{equation}
Consequently, positive reward scaling preserves each agent's preference
ordering over fixed joint policies, its best-response correspondence, and the
set of Nash equilibria.  It also preserves every homogeneous linear incentive
comparison whose two sides consist only of discounted-return terms.
\end{suppTheorem}

\begin{proof}
Linearity of expectation gives
\[
J_i^{(s)}(\boldsymbol{\pi})
=\mathbb{E}_{\boldsymbol{\pi}}
\left[\sum_{t\geq0}\gamma^t s r_{i,t}\right]
=sJ_i(\boldsymbol{\pi}).
\]
Because $s>0$, for any $\boldsymbol{\pi}$ and
$\boldsymbol{\pi}'$,
$J_i(\boldsymbol{\pi})\geq J_i(\boldsymbol{\pi}')$ if and only if
$J_i^{(s)}(\boldsymbol{\pi})\geq
J_i^{(s)}(\boldsymbol{\pi}')$.  Taking an $\arg\max$ over the policy of agent
$i$ therefore gives the same best-response set, and the fixed points of the
best-response correspondence, hence the Nash-equilibrium set, are unchanged.
Finally, multiplying every return term on both sides of a homogeneous linear
comparison by the same positive constant preserves its truth value.
\end{proof}

For the repeated production-and-exchange game, the same statement can be seen
directly from the incentive quantities.

\begin{suppCorollary}[Invariance of the repeated-game incentive threshold]
\label{cor:minex-threshold}
Let $P$, $R$, and $T$ denote the punishment, cooperative, and one-shot
deviation returns, with $T>P$ and
$\delta^\star=(T-R)/(T-P)$.  Under positive reward scaling,
$P^{(s)}=sP$, $R^{(s)}=sR$, and $T^{(s)}=sT$, so
\begin{equation}
\delta^{\star,(s)}
=\frac{T^{(s)}-R^{(s)}}{T^{(s)}-P^{(s)}}
=\delta^\star.
\end{equation}
Thus, profitability and incentive compatibility are preserved even though the
learning targets are rescaled.
\end{suppCorollary}

\begin{proof}
Positive reward scaling gives $P^{(s)}=sP$, $R^{(s)}=sR$, and
$T^{(s)}=sT$. Substitution into the repeated-game threshold cancels the
common positive factor, yielding $\delta^{\star,(s)}=\delta^\star$.
\end{proof}

Theorem~\ref{prop:strategic-invariance} concerns the environment game.  It
does not imply that a neural optimizer applied to policy and value losses is
scale invariant.

\subsection{Matched Fixed-Batch Gradient Decomposition}

Let $\theta$ parameterize the actor torso, $\psi$ the policy head, and $\phi$
the value head.  Consider a common checkpoint and a fixed batch
$\mathcal{B}$ of observations and actions.  At reward scale $s$, write
$\widetilde A_b^{(s)}$ for the advantage used by the policy loss and
$y_b^{(s)}$ for the value target of sample $b$.  We use the quadratic value
loss
\begin{equation}
L_V^{(s)}(\theta,\phi)
=\frac{1}{2|\mathcal{B}|}
\sum_{b\in\mathcal{B}}
\left(V_{\theta,\phi}(o_b)-y_b^{(s)}\right)^2.
\label{eq:value-loss-general}
\end{equation}
The factor $1/2$ is conventional and only removes an inessential factor of two
from the gradients.

The fixed-batch comparison below uses two explicit conditions.  First, the
parameters, observations, actions, and old-policy probabilities are identical
across scales at the point of comparison.  Second, the normalized policy
advantages are identical,
\begin{equation}
\widetilde A_b^{(s)}=\widetilde A_b^{(1)}
\quad\text{for every }b\in\mathcal{B}.
\label{eq:policy-invariance-condition}
\end{equation}
Condition~\eqref{eq:policy-invariance-condition} holds, for example, when the
raw advantages satisfy $\widehat A_b^{(s)}=s\widehat A_b^{(1)}$ and exact
batch standardization is used without a scale-dependent additive constant.
Under this condition, the PPO policy-gradient and entropy-gradient terms are
the same at the matched checkpoint.  Define their common torso gradient by
$g_{\pi,\theta}$.

To avoid assuming more about the target construction than the implementation
provides, define the target displacement
\begin{equation}
d_b^{(s)}:=y_b^{(s)}-y_b^{(1)}.
\label{eq:target-displacement}
\end{equation}
Strict proportional scaling, $y_b^{(s)}=s y_b^{(1)}$, is treated as a special
case rather than built into the general result.

\begin{suppLemma}[Value-gradient displacement]
\label{lem:value-gradient-displacement}
At a common parameter point $(\theta,\phi)$,
\begin{equation}
\nabla_\theta L_V^{(s)}-\nabla_\theta L_V^{(1)}
=-\frac{1}{|\mathcal{B}|}
\sum_{b\in\mathcal{B}}
 d_b^{(s)}\nabla_\theta V_{\theta,\phi}(o_b).
\label{eq:value-gradient-displacement}
\end{equation}
\end{suppLemma}

\begin{proof}
Differentiating Equation~\eqref{eq:value-loss-general} gives
\[
\nabla_\theta L_V^{(s)}
=\frac{1}{|\mathcal{B}|}
\sum_{b\in\mathcal{B}}
\left(V_{\theta,\phi}(o_b)-y_b^{(s)}\right)
\nabla_\theta V_{\theta,\phi}(o_b).
\]
Subtracting the scale-$1$ expression cancels the prediction term
$V_{\theta,\phi}(o_b)\nabla_\theta V_{\theta,\phi}(o_b)$ sample by sample,
leaving Equation~\eqref{eq:value-gradient-displacement}.
\end{proof}

Let $\kappa_{\mathrm{X0}}=1$ and $\kappa_{\mathrm{X1}}=0$ indicate whether the
value loss is allowed to update the actor torso.  Before clipping or adaptive
preconditioning, one gradient-descent step with learning rate
$\alpha_{\mathrm{lr}}>0$ is
\begin{equation}
\Delta\theta_{a,s}
:=\theta_{a,s}^{+}-\theta
=-\alpha_{\mathrm{lr}}
\left(g_{\pi,\theta}
+\kappa_a c_V\nabla_\theta L_V^{(s)}\right).
\label{eq:arm-update-unclipped}
\end{equation}

\begin{suppTheorem}[Direct-route scale dependence]
\label{prop:direct-route}
Under the matched fixed-batch conditions above and the unclipped update in
Equation~\eqref{eq:arm-update-unclipped}, the same-scale routing contrast is
\begin{equation}
\Delta\theta_{\mathrm{X0},s}
-\Delta\theta_{\mathrm{X1},s}
=-\alpha_{\mathrm{lr}}c_V\nabla_\theta L_V^{(s)}.
\label{eq:same-scale-route-contrast}
\end{equation}
Moreover, the arm-by-scale difference-in-differences is
\begin{align}
&\left(\Delta\theta_{\mathrm{X0},s}
-\Delta\theta_{\mathrm{X0},1}\right)
-\left(\Delta\theta_{\mathrm{X1},s}
-\Delta\theta_{\mathrm{X1},1}\right)
\nonumber\\
&\qquad=
\frac{\alpha_{\mathrm{lr}}c_V}{|\mathcal{B}|}
\sum_{b\in\mathcal{B}}
 d_b^{(s)}\nabla_\theta V_{\theta,\phi}(o_b).
\label{eq:direct-did-general}
\end{align}
If the value targets scale proportionally,
$y_b^{(s)}=s y_b^{(1)}$, then
\begin{equation}
\eqref{eq:direct-did-general}
=\alpha_{\mathrm{lr}}c_V(s-1)
\,\mathbb{E}_{b\sim\mathcal{B}}
\left[y_b^{(1)}\nabla_\theta V_{\theta,\phi}(o_b)\right].
\label{eq:direct-did-proportional}
\end{equation}
The direct term is identically zero in X1 and is absent in X5, which has no
learned value objective.
\end{suppTheorem}

\begin{proof}
Equation~\eqref{eq:same-scale-route-contrast} follows by subtracting the X1
update from the X0 update at the same scale.  Subtracting the scale-$1$ update
from the scale-$s$ update in each arm cancels the common policy term.  The X1
value-to-torso term is zero by construction, while the X0 difference is
$-\alpha_{\mathrm{lr}}c_V
(\nabla_\theta L_V^{(s)}-\nabla_\theta L_V^{(1)})$.
Applying Lemma~\ref{lem:value-gradient-displacement} proves
Equation~\eqref{eq:direct-did-general}; substituting
$d_b^{(s)}=(s-1)y_b^{(1)}$ gives
Equation~\eqref{eq:direct-did-proportional}.
\end{proof}

\begin{suppCorollary}[Scale invariance of the stopped route under unclipped SGD]
\label{cor:x1-unclipped}
Under the same fixed-batch conditions,
\begin{equation}
\Delta\theta_{\mathrm{X1},s}
=\Delta\theta_{\mathrm{X1},1}.
\end{equation}
Thus, in an unclipped SGD comparison, any scale-dependent X0--X1 separation at
the first matched update is attributable to the direct value-to-torso route,
subject to the stated policy-invariance condition.  X5 also lacks this direct
term, but its Monte-Carlo policy estimator is not assumed to be identical to
the X0/X1 estimator.
\end{suppCorollary}

\begin{proof}
For X1, the routing indicator is zero, so the torso update contains only the
matched policy gradient. The policy-invariance condition makes that gradient
identical across scales, which gives the stated equality.
\end{proof}

\subsection{Indirect Coupling Through Global Gradient Clipping}

Stopping the value gradient at the torso does not necessarily make the actor
update independent of reward scale when a single global norm is used to clip
all trainable parameters.  Let $G>0$ be the clipping threshold and define
\begin{equation}
\operatorname{Clip}_G(g)
:=\chi(g)g,
\qquad
\chi(g):=1\wedge\frac{G}{\lVert g\rVert_2}.
\label{eq:global-clip-operator}
\end{equation}
For X1, let the full pre-clipping gradient be
\begin{equation}
g_{\mathrm{X1},s}
=\left(g_{\pi,\theta},\;g_{\pi,\psi},\;
 c_V g_{V,\phi}^{(s)}\right),
\label{eq:x1-full-gradient}
\end{equation}
where the value-head gradient $g_{V,\phi}^{(s)}$ can depend on reward scale
even though the value-to-torso gradient is stopped.

\begin{suppProposition}[Clipping-mediated scale dependence]
\label{prop:clipping-coupling}
Under global norm clipping, the X1 torso update is
\begin{equation}
\Delta\theta_{\mathrm{X1},s}^{\mathrm{clip}}
=-\alpha_{\mathrm{lr}}\chi_s g_{\pi,\theta},
\qquad
\chi_s:=1\wedge\frac{G}{\lVert g_{\mathrm{X1},s}\rVert_2}.
\label{eq:x1-clipped-update}
\end{equation}
Hence
\begin{equation}
\Delta\theta_{\mathrm{X1},s}^{\mathrm{clip}}
-\Delta\theta_{\mathrm{X1},1}^{\mathrm{clip}}
=-\alpha_{\mathrm{lr}}(\chi_s-\chi_1)g_{\pi,\theta}.
\label{eq:x1-clipped-difference}
\end{equation}
Therefore, X1 can remain scale dependent even with an exactly zero
value-to-torso gradient whenever reward scaling changes the full gradient norm
and at least one comparison lies in the clipping regime.
\end{suppProposition}

\begin{proof}
The clipping operator multiplies every component of the full gradient by the
same scalar $\chi_s$.  Selecting the torso coordinates of
$-\alpha_{\mathrm{lr}}\operatorname{Clip}_G(g_{\mathrm{X1},s})$ yields
Equation~\eqref{eq:x1-clipped-update}; subtracting the scale-$1$ expression
yields Equation~\eqref{eq:x1-clipped-difference}.
\end{proof}

In X0, the scale-dependent value gradient changes both the torso direction and
the clipping multiplier, so direct routing and clipping-mediated coupling are
entangled.  This is why the unclipped-SGD cell provides the cleanest
fixed-batch X0/X1 interpretation, whereas the clipped-Adam cell tests the route
inside a more realistic optimizer system.

\begin{suppRemark}[Adaptive-optimizer dependence]
For Adam,
$m_t=\beta_1m_{t-1}+(1-\beta_1)g_t$ and
$v_t=\beta_2v_{t-1}+(1-\beta_2)g_t^{\odot2}$, followed by a coordinatewise
preconditioned step.  Exact cancellation of a positive scale factor requires a
much stronger condition than reward scaling: the complete gradient history
must be related by a common multiplicative factor across all relevant
coordinates, with matched moments and negligible numerical offset.  Here,
reward scaling perturbs only value-related components and generally does so
additively through the target residual.  Consequently, neither the first-step
direction nor the later Adam trajectory is generally scale invariant.
\end{suppRemark}

\subsection{From Parameter Displacement to Maintenance Risk}

The previous results identify a parameter-update route, not a theorem that
cooperation must fail.  To connect the update to behavior, let
$h(\theta,\psi)$ be a differentiable local margin for one active cooperation
criterion, with $h\geq0$ indicating success.  A first-order expansion around
the matched checkpoint gives the following local implication.

\begin{suppCorollary}[First-order maintenance-margin effect]
\label{cor:behavioral-margin}
Under the assumptions of Theorem~\ref{prop:direct-route}, fix any common
policy head $\bar\psi$ (including the common post-update head $\psi^+$) and
suppose $h(\cdot,\bar\psi)$ is twice continuously differentiable on the local
update neighborhood, with
$\lVert\nabla_{\theta}^{2}h\rVert_{\mathrm{op}}\leq L_h$.  Then
\begin{align}
&\Big(h(\theta+\Delta\theta_{\mathrm{X0},s},\bar\psi)
-h(\theta+\Delta\theta_{\mathrm{X0},1},\bar\psi)\Big)
\nonumber\\[-2pt]
&\quad-
\Big(h(\theta+\Delta\theta_{\mathrm{X1},s},\bar\psi)
-h(\theta+\Delta\theta_{\mathrm{X1},1},\bar\psi)\Big)
\nonumber\\
&\qquad=
\frac{\alpha_{\mathrm{lr}}c_V}{|\mathcal{B}|}
\nabla_\theta h(\theta,\bar\psi)^\top
\sum_{b\in\mathcal{B}}
 d_b^{(s)}\nabla_\theta V_{\theta,\phi}(o_b)
+R_h,
\label{eq:margin-first-order}
\end{align}
where $\varepsilon_\Delta$ is the largest update norm in the comparison and
$|R_h|\leq2L_h\varepsilon_\Delta^2$.  The direct route reduces the local
cooperation margin only when the inner product in
Equation~\eqref{eq:margin-first-order} is sufficiently negative relative to
this remainder; it can be harmless or margin-increasing when the alignment is
zero or positive.
\end{suppCorollary}

\begin{proof}
Apply Taylor's theorem in $\theta$ at fixed $\bar\psi$ to the four margin terms.
Their constant terms cancel, and their linear terms equal
$\nabla_\theta h(\theta,\bar\psi)$ dotted with the update
difference-in-differences.  Theorem~\ref{prop:direct-route} gives the displayed
first-order term.  Each Taylor remainder is at most
$L_h\varepsilon_\Delta^2/2$ in absolute value; the triangle inequality over
four terms gives the stated bound.
\end{proof}

Corollary~\ref{cor:behavioral-margin} makes the scope of the theory explicit.
Reward scale changes the direct X0 update, but a maintenance event additionally
requires adverse alignment with a behavioral boundary and sufficient movement
to cross the available margin.  Later on-policy batches, optimizer state, and
recurrent dynamics can amplify or reverse the one-step effect.  The theory
therefore motivates, rather than replaces, the optimizer--environment boundary
matrix evaluated experimentally.

\end{document}